\documentclass{article}
\usepackage{spconf,amsmath,graphicx}

\usepackage{cite}

\usepackage{amssymb}
\usepackage{amsthm}

\usepackage{cite}
\usepackage{array}

\usepackage{url,color}

\def\kron{\otimes}
\def\tr{\mathrm{tr}}

\def\diag{\mathrm{diag}}

\def\Htran{\mbox{\tiny H}}
\def\Ttran{\mbox{\tiny T}}

\newcommand{\fracSum}[1]{{\underset{{#1}}{\sum}}}

\newcommand{\fracSumtwo}[2]{\overset{#2}{\underset{#1}{\sum}}}
\newcommand{\vect}[1]{\mathbf{#1}}

\theoremstyle{remark}

\newtheorem{theorem}{Theorem}
\newtheorem{corollary}{Corollary}
\newtheorem{lemma}{Lemma}

\title{Massive MIMO Systems with Hardware-Constrained Base Stations}

\name{Emil~Bj\"ornson$^{\star \dagger}$ \qquad Michail Matthaiou$^{\ddagger}$ \qquad M\'erouane~Debbah$^{\star}$\thanks{E.~Bj\"ornson is funded by the International Postdoc Grant 2012-228 from the Swedish Research Council. This research has been supported by the ERC Starting Grant 305123 MORE (Advanced Mathematical Tools for Complex Network Engineering).}}

\address{$^\star$Alcatel-Lucent Chair on Flexible Radio, SUPELEC, Gif-sur-Yvette, France\\
$^\dagger$ACCESS Centre, Dept.~of Signal Processing, KTH Royal Institute of Technology, Stockholm, Sweden\\
$^\ddagger$ECIT Institute, Queen's University Belfast, U.K. and S2, Chalmers University of Technology, Sweden}

\begin{document}
\ninept

\maketitle

\begin{abstract}
Massive multiple-input multiple-output (MIMO) systems are cellular networks where the base stations (BSs) are equipped with unconventionally many antennas. Such large antenna arrays offer huge spatial degrees-of-freedom for transmission optimization; in particular, great signal gains, resilience to imperfect channel knowledge, and small inter-user interference are all achievable without extensive inter-cell coordination. The key to cost-efficient deployment of large arrays is the use of hardware-constrained base stations with low-cost antenna elements, as compared to today's expensive and power-hungry BSs. Low-cost transceivers are prone to hardware imperfections, but it has been conjectured that the excessive degrees-of-freedom of massive MIMO would bring robustness to such imperfections. We herein prove this claim for an uplink channel with multiplicative phase-drift, additive distortion noise, and noise amplification. Specifically, we derive a closed-form scaling law that shows how fast the imperfections increase with the number of antennas.
\end{abstract}


\vspace*{-1mm}

\section{Introduction}

\vspace*{-1mm}

Massive densification, in terms of more antennas per unit area, is a key enabler to higher area throughput in future wireless networks \cite{Hoydis2013c}. This is achieved using multi-user MIMO techniques, by adding more antennas to the macro BSs and/or distributing the antennas by ultra-dense deployment of small BSs. These approaches are non-conflicting, since the former operates in current frequency bands and the latter is expected to primarily operate in new mm-wave bands \cite{Baldemair2013a}.

This paper considers the former densification approach, which was first proposed in \cite{Marzetta2010a} and is nowadays commonly referred to as \emph{massive MIMO} \cite{Hoydis2013a,Larsson2014a}. The main characteristics of massive MIMO are that the BSs are equipped with large antenna arrays with hundreds (or even thousands) of antennas, which are used to serve tens (or even hundreds) of users. In other words, the number of antennas, $N$, and number of users per BS, $K$, are unconventionally large, but can differ by an order of magnitude. For this reason, massive MIMO brings unprecedented spatial degrees-of-freedom, which enable precoding with strong signal gains, give near-orthogonal user channels, and resilience to imperfect channel knowledge \cite{Rusek2013a}.

Apart from achieving high area throughput, recent works have investigated additional ways to capitalize on the huge degrees-of-freedom offered by massive MIMO. Towards this end, \cite{Hoydis2013c} showed that massive MIMO enables implicit coordination between systems that operate in the same band. Moreover, it was shown in \cite{Hoydis2013a} and \cite{Ngo2013a} that the transmit powers can be reduced as $1/\sqrt{N}$ with only a minor loss in throughput. This offers major reductions in the emitted power, but is actually bad from an energy efficiency (EE) perspective---the EE is maximized by increasing the transmit power with $N$ to compensate for the increasing circuit power \cite{Bjornson2014b}.

The work presented here explores whether the huge degrees-of-freedom offered by massive MIMO provide robustness to transceiver hardware imperfections; for example, phase drifts, quantization errors, and noise amplification. Robustness to hardware imperfections has been conjectured in overview articles, such as \cite{Larsson2014a}, and is notably important since the deployment cost of massive MIMO scales linearly with $N$ unless we resort to using cheaper hardware with larger imperfections. Constant envelope precoding was analyzed in \cite{Mohammed2013a} to facilitate the use of power-efficient amplifiers, while the impact of phase drifts was analyzed and simulated in \cite{Pitarokoilis2012a,Pitarokoilis2014a}. A preliminary proof of the conjecture was provided in \cite{Bjornson2014a}, but the authors considered only additive distortions and, thus, ignored other characteristics of hardware imperfections. It was shown that the distortion variance can increase as $\sqrt{N}$ with only minor throughput losses.

In this work, we consider an uplink massive MIMO system with hardware imperfections that cause phase drifts, additive distortions, and noise amplification; this model is more general compared to \cite{Pitarokoilis2012a,Pitarokoilis2014a,Bjornson2014a} which investigated merely one of these effects. We derive a new linear minimum mean square error (LMMSE) channel estimator and closed-form achievable user rates. Based on the analytical results, we prove the conjecture by obtaining intuitive scaling laws that show how fast the hardware imperfections can be increased with $N$. The results are validated numerically in a realistic simulation setup, while the impact on circuit design is considered in \cite{Bjornson2014c}.

\vspace*{-1mm}

\section{System Model}

\vspace*{-1mm}

This paper considers the uplink of a cellular network with $L\geq 1$ cells. Each cell consists of $K$ single-antenna user equipments (UEs) that communicate simultaneously with a BS which is equipped with an array of $N$ antennas. Our analysis holds for any $N$ and $K$, but we are primarily interested in massive MIMO topologies where $N \gg K \gg 1$. The channel from UE $k$ in cell $l$ to BS $j$ is denoted as $\vect{h}_{jlk} = [h_{jlk}^{(1)} \, \ldots \, h_{jlk}^{(N)}]^{\Ttran} \in \mathbb{C}^N$ and is modeled as Rayleigh block fading; thus, it takes a static realization for a coherence block of $T$ channel uses and independent realizations between blocks. Each channel is circularly symmetric complex Gaussian distributed with zero mean and covariance matrix $\lambda_{jlk} \vect{I}_N$: $\vect{h}_{jlk} \sim \mathcal{CN}(\vect{0},\lambda_{jlk} \vect{I}_N)$.\footnote{The assumption of independent fading implies that the array dimensions grow with $N$ to keep the inter-antenna distance sufficiently large. However, the analysis in this paper can be easily extended to spatially correlated channels as in \cite{Hoydis2013a} and \cite{Bjornson2014a}, but at the cost of complicating the notation and results.} The average channel attenuation $\lambda_{jlk}>0$ is different for each combinations of BS and UE and depends, for example, on the distance.

The received signal $\vect{y}_j(t) \in \mathbb{C}^{N}$ at BS $j$ at a given channel use $t \in \{ 1,\ldots,T \}$ in the coherence block is conventionally modeled as \cite{Marzetta2010a,Rusek2013a,Hoydis2013a,Ngo2013a}
\begin{equation} \label{eq:conventional-model}
\vect{y}_j(t) = \sum_{l=1}^{L} \vect{H}_{jl} \vect{x}_{l}(t) + \vect{n}_{j}(t)
\end{equation}
where the transmit signal in cell $l$ is $\vect{x}_{l}(t) = [x_{l1}(t) \, \ldots \, x_{lK}(t)]^{\Ttran} \in \mathbb{C}^{K}$ and $\vect{H}_{jl} = [ \vect{h}_{jl1} \, \ldots \, \vect{h}_{jlK}] \in \mathbb{C}^{N \times K}$. The signal $x_{lk}(t)$ sent by UE $k$ in cell $l$ at channel use $t$ is either a deterministic pilot symbol (used for channel estimation) or an information symbol from a Gaussian codebook; in any case, the expectation of the transmit power is bounded as $\mathbb{E}\{ |x_{lk}(t) |^2 \} \leq p_{lk}$. The thermal noise vector $\vect{n}_{j}(t) \sim \mathcal{CN}(\vect{0},\sigma^2 \vect{I}_N)$ is independent in time and has variance $\sigma^2$.

The conventional model in \eqref{eq:conventional-model} is well-accepted for small-scale MIMO systems, but has an important drawback when applied to massive MIMO: it assumes that the BS array consists of $N$ high-quality antenna elements which are all fully synchronized. Consequently, the cost and circuit power consumption would \emph{at least} grow linearly with $N$, thus making the deployment of massive MIMO rather questionable from an overall cost and efficiency perspective.

In this paper, we analyze the far more realistic scenario of hardware-constrained BSs. Specifically, each BS has hardware imperfections that distort the communication in three ways: 1) received signals are shifted in phase; 2) distortion noise is added with a power proportional to the total received signal power; and 3) amplification of the thermal noise. In this generalized scenario, the received signal at BS $j$ at a given channel use $t \in \{ 1,\ldots,T \}$  is modeled as
\begin{equation} \label{eq:generalized-model}
\vect{y}_j(t) = \vect{D}_{\boldsymbol{\phi}_j(t)} \sum_{l=1}^{L} \vect{H}_{jl} \vect{x}_{l}(t) + \boldsymbol{\upsilon}_{j}(t) + \boldsymbol{\eta}_{j}(t)
\end{equation}
where the channel matrices $\vect{H}_{jl}$ and transmitted signals $\vect{x}_{l}(t)$ are exactly as in \eqref{eq:conventional-model}.
The hardware imperfections are characterized by:
\begin{enumerate}

\item The phase-drift matrix $\vect{D}_{\boldsymbol{\phi}_j(t)} \! \triangleq \! \diag(e^{\imath \phi_{j1}(t)},\ldots, e^{\imath \phi_{jN}(t)})$ where $\phi_{jn}(t)$ is the phase drift at the $n$th antenna of BS $j$ at time $t$. It follows a Wiener process $\phi_{jn}(t) \!\sim\! \mathcal{N}( \phi_{jn}(t\!-\!1), \delta)$; thus, $\phi_{jn}(t)$ equals $\phi_{jn}(t-1)$ plus an independent Gaussian innovation of variance $\delta$. Each antenna experiences an independent phase-drift process with the same variance (e.g., due to the use of separate oscillators with identical properties).

\item The distortion noise $\boldsymbol{\upsilon}_{j}(t) \sim \mathcal{CN}(\vect{0},\vect{\Upsilon}_j(t) )$ where $\vect{\Upsilon}_j(t) \triangleq \kappa^2 \sum_{l=1}^{L} \sum_{k=1}^{K} \mathbb{E}\{ |x_{lk}(t) |^2 \} \diag( |h_{jlk}^{(1)}|^2,\ldots, |h_{jlk}^{(N)}|^2)$ for a given channel realization. The distortion noise is thus independent between antennas and channel uses, and the variance at a given antenna is proportional to the current received signal power at this antenna. The proportionality parameter $\kappa \geq 0$ is the error vector magnitude (EVM), which is a common quality measure of transceiver hardware \cite{Holma2011a}.

\item The receiver noise $\boldsymbol{\eta}_{j}(t) \!=\! \sqrt{\xi} \vect{n}_{j}(t)  \sim \mathcal{CN}(\vect{0},\sigma^2 \xi \vect{I}_N)$ where the parameter $\xi \geq 1$ is the noise amplification factor.
\end{enumerate}

The generalized system model in \eqref{eq:generalized-model} is based on \cite{Schenk2008a,Mezghani2010a,wenk2010mimo,Bjornson2013d} and characterizes the joint behavior of different types of hardware imperfections at the BSs. For example, phase noise in the oscillators causes phase-drifts, finite-resolution analog-to-digital converters cause distortion noise, and the electronic BS amplifier causes noise amplification. The distortions either originate from uncalibrated hardware imperfections or residual errors after calibration.

We will derive a channel estimator and achievable user rates for the system model in \eqref{eq:generalized-model}. By analyzing the performance as $N\rightarrow\infty$, we will bring insights into the fundamental impact of the parameters $\delta$, $\kappa$, and $\xi$, which characterize the BS hardware imperfections.

\section{Performance Analysis}
\label{sec:performance-analysis}

In this section, we derive achievable user rates for the uplink system in \eqref{eq:generalized-model} and scaling laws for how quickly the BS hardware imperfections can be increased with $N$ and still achieve non-zero rates.

\subsection{Channel Estimation}

The achievable rates are computed under the assumption that the first $B \geq K$ channel uses of each coherence block are dedicated for pilot-based channel estimation. UE $k$ in cell $j$ transmits a predefined pilot sequence $\tilde{\vect{x}}_{jk} = [x_{jk}(1) \, \ldots \, x_{jk}(B)]^{\Ttran} \in \mathbb{C}^{B}$. The pilot sequences are selected arbitrarily under the above mentioned power constraints and our analysis supports any choice. However, it is reasonable to make the sequences $\tilde{\vect{x}}_{j1},\ldots,\tilde{\vect{x}}_{jK}$ in cell $j$ linearly independent to avoid unnecessary intra-cell interference. Due to the limited coherence block length $T$, inter-cell interference is often unavoidable but the pilot sequences can also be designed and allocated to also reduce  inter-cell interference \cite{Yin2013a}.

For any given set of pilot sequences, we now derive estimators of the effective channels $\vect{h}_{jlk}(t) \triangleq \vect{D}_{\boldsymbol{\phi}_j(t)} \vect{h}_{jlk}$ at any channel use $t \geq B$ for all $j,l,k$. The conventional multi-antenna channel estimators from \cite{Kay1993a,Kotecha2004a,Bjornson2010a} cannot be applied in this work since the generalized system model in \eqref{eq:generalized-model} has two non-standard properties: The pilot transmission is corrupted by random phase-drifts and the distortion noise is statistically dependent on the channels. Therefore, we derive a new LMMSE estimator for the system model at hand.

\begin{theorem} \label{theorem:LMMSE-estimation}
Let $\boldsymbol{\psi}_j = [\vect{y}_j^{\Ttran}(1) \, \ldots \, \vect{y}_j^{\Ttran}(B)]^{\Ttran} \in \mathbb{C}^{NB}$ denote the received signal at BS $j$ from the pilot transmission. The LMMSE estimate of $\vect{h}_{jlk}(t)$ at any channel use $t\geq B$ for any $l$ and $k$ is
\begin{equation} \label{eq:LMMSE-estimator}
  \hat{\vect{h}}_{jlk}(t) = \left( \lambda_{jlk} \tilde{\vect{x}}_{lk}^{\Htran}  \vect{D}_{\boldsymbol{\delta}(t)} \boldsymbol{\Psi}^{-1}_j \kron \vect{I}_N \right) \boldsymbol{\psi}_j
\end{equation}
where $\vect{D}_{\boldsymbol{\delta}(t)} \triangleq \diag( e^{-\frac{\delta}{2} (t-1)}, e^{-\frac{\delta}{2} (t-2)}, \ldots, e^{-\frac{\delta}{2} (t-B)})$,
\begin{equation}
\boldsymbol{\Psi}_j \triangleq \sum_{\ell=1}^{L} \sum_{m=1}^{K} \lambda_{j \ell m} \vect{X}_{\ell m} + \sigma^2 \xi \vect{I}_B,
\end{equation}
$\kron$ is the Kronecker product, and element $(i_1,i_2)$ of $\vect{X}_{\ell m} \!\in\! \mathbb{C}^{B \times B}$ is
\begin{equation}
[\vect{X}_{\ell m} ]_{i_1,i_2} = \begin{cases} |x_{\ell m}(i_1)|^2 (1\!+\!\kappa^2) , & i_1 = i_2, \\ x_{\ell m}(i_1) x_{\ell m}^*(i_2) e^{-\frac{\delta}{2} |i_1-i_2|}, &  i_1 \neq i_2. \end{cases}
\end{equation}
The corresponding error covariance matrix is
\begin{equation} \label{eq:LMMSE-error-cov}
\begin{split}
\vect{C}_{jlk} &= \mathbb{E}\left\{ ( \vect{h}_{jlk}(t) -\hat{\vect{h}}_{jlk}(t) )( \vect{h}_{jlk}(t) - \hat{\vect{h}}_{jlk}(t) )^{\Htran}   \right\} \\ &= \lambda_{jlk} \left( 1 - \lambda_{jlk} \tilde{\vect{x}}_{lk}^{\Htran}  \vect{D}_{\boldsymbol{\delta}(t)} \boldsymbol{\Psi}^{-1}_j \vect{D}_{\boldsymbol{\delta}(t)}^{\Htran} \tilde{\vect{x}}_{lk} \right) \vect{I}_N
\end{split}
\end{equation}
and the mean-squared error (MSE) becomes $\mathrm{MSE}_{jlk} = \tr( \vect{C}_{jlk})$.
\end{theorem}
\begin{proof}
The general expression for an LMMSE estimator is $\hat{\vect{h}}_{jlk}(t) = \mathbb{E}\{ \vect{h}_{jlk}(t) \boldsymbol{\psi}_j^{\Htran} \}  \left( \mathbb{E}\{ \boldsymbol{\psi}_j \boldsymbol{\psi}_j^{\Htran} \}  \right)^{-1}  \boldsymbol{\psi}_j$ \cite[Chapter 12]{Kay1993a}. The theorem follows from algebraic computation of the two expectations.
\end{proof}

Although the channels are block fading, the phase-drifts caused by hardware imperfections make the effective channels $\vect{h}_{jlk}(t)$ change between channel uses. The new LMMSE estimator in Theorem \ref{theorem:LMMSE-estimation} predicts the effective channel for each $t \in \{B+1,\ldots,T\}$ during the data transmission. Next, we use these predictors to design receive filters and derive the corresponding achievable user rates.

\begin{figure*}[t!]
\begin{align} \label{eq:achievable-SINR}
\mathrm{SINR}_{jk}(t) =  \frac{ p_{jk} | \mathbb{E}\{ \vect{v}_{jk}^{\Htran}(t) \vect{h}_{jjk}(t) \} |^2 }{ \fracSumtwo{l=1}{L} \fracSumtwo{m=1}{K} p_{lm}  \mathbb{E}\{ |\vect{v}_{jk}^{\Htran}(t) \vect{h}_{jlm}(t) |^2  \} - p_{jk} | \mathbb{E}\{ \vect{v}_{jk}^{\Htran}(t) \vect{h}_{jjk}(t) \} |^2 + \mathbb{E}\{ |\vect{v}_{jk}^{\Htran}(t) \boldsymbol{\upsilon}_j(t) |^2  \}   + \sigma^2 \xi \mathbb{E}\{ \| \vect{v}_{jk}(t) \|^2\}  } \tag{8}
\end{align} \vskip-2mm
\hrulefill
\vskip-5mm
\end{figure*}

\subsection{Achievable User Rates}
\label{subsec:user-rates}

Achievable user rates for the generalized uplink channel in \eqref{eq:generalized-model} are given in the next lemma. These form a base for asymptotic analysis.

\begin{lemma} \label{lemma:achievable-rates}
Suppose BS $j$ has statistical channel knowledge and applies the filters $\vect{v}_{jk}^{\Htran}(t) \in \mathbb{C}^{N}$, $t=B+1,\ldots,T$, to receive the signals from its $k$th UE, then an ergodic achievable user rate is \vskip-4mm
\begin{equation} \label{eq:achievable-rate}
R_{jk} = \frac{1}{T}  \sum_{t=B+1}^{T} \log_2 \left( 1 + \mathrm{SINR}_{jk}(t) \right) \quad [\textrm{bit/channel use}]
\end{equation}
where $\mathrm{SINR}_{jk}(t)$ is given in \eqref{eq:achievable-SINR} at the top of this page and all UEs transmit with full power (i.e., $\mathbb{E}\{ |x_{lk}(t) |^2 \} = p_{lk}$ for all $l,k$).
\end{lemma}
\begin{proof}
Since the effective channels vary with $t$, we follow the approach in \cite{Pitarokoilis2012a,Pitarokoilis2014a} and compute an achievable rate for each $t$. We obtain \eqref{eq:achievable-rate} by averaging over the coherence block. The SINR in \eqref{eq:achievable-SINR} is obtained by treating inter-user interference and additive distortions as Gaussian noise (a worst-case assumption \cite{Hassibi2003a}) and only exploiting knowledge of the average effective channel $\mathbb{E}\{ \vect{v}_{jk}^{\Htran}(t) \vect{h}_{jjk}(t) \}$ while any deviation is treated as worst-case Gaussian noise \cite{Medard2000a,Marzetta2010a}.
\end{proof}

\setcounter{equation}{8}

The rate expressions in Lemma \ref{lemma:achievable-rates} can be utilized for any choice of receive filters. The next theorem gives closed-form expressions for all expectations under maximum ratio combining (MRC).

\begin{theorem} \label{theorem:MRC-expectations}
If the MRC filter $\vect{v}_{jk}(t) =  \hat{\vect{h}}_{jjk}(t)$ is used, then
\begin{align}
\mathbb{E}\{ \|\vect{v}_{jk}(t) \|^2\} & = N \lambda_{jjk}^2 \tilde{\vect{x}}_{jk}^{\Htran}  \vect{D}_{\boldsymbol{\delta}(t)} \boldsymbol{\Psi}^{-1}_j \vect{D}_{\boldsymbol{\delta}(t)}^{\Htran} \tilde{\vect{x}}_{jk}  \label{eq:MRC-squared-norm} \\
\mathbb{E}\{ \vect{v}_{jk}^{\Htran}(t) \vect{h}_{jjk}(t) \} & = \mathbb{E}\{ \|\vect{v}_{jk}(t) \|^2\}  \label{eq:MRC-first-moment} \\
\mathbb{E}\{ | \vect{v}_{jk}^{\Htran}(t) \vect{h}_{jlm}(t) |^2 \} &= \lambda_{jlm} \mathbb{E}\{ \| \vect{v}_{jk}(t) \|^2\}  \notag  \\
&\!\!\!\!\!\!\!\!\!\!\!\!\!\!\!\!\!\!\!\!\!\!\!\!\!\!\!\!\!  + N  \lambda_{jjk}^2 \lambda_{jlm}^2 \tilde{\vect{x}}_{jk}^{\Htran}   \vect{D}_{\boldsymbol{\delta}(t)} \boldsymbol{\Psi}_j^{-1} \vect{X}_{lm}  \boldsymbol{\Psi}_j^{-1} \vect{D}_{\boldsymbol{\delta}(t)}^{\Htran} \tilde{\vect{x}}_{jk} \notag \\
&\!\!\!\!\!\!\!\!\!\!\!\!\!\!\!\!\!\!\!\!\!\!\!\!\!\!\!\!\!  + N(N-1)  \lambda_{jjk}^2 \lambda_{jlm}^2 |\tilde{\vect{x}}_{jk}^{\Htran}  \vect{D}_{\boldsymbol{\delta}(t)} \boldsymbol{\Psi}_j^{-1}  \vect{D}_{\boldsymbol{\delta}(t)}^{\Htran} \tilde{\vect{x}}_{lm} |^2 \label{eq:MRC-second-moment} \\
\mathbb{E}\{ |\vect{v}_{jk}^{\Htran}(t) \boldsymbol{\upsilon}_j(t) |^2  \}&= \kappa^2 \sum_{l=1}^{L} \sum_{m=1}^{K} p_{lm} \lambda_{jlm} \mathbb{E}\{ \|\vect{v}_{jk}(t) \|^2\} \label{eq:MRC-cross-moment}  \\
&\!\!\!\!\!\!\!\!\!\!\!\!\!\!\!\!\!\!\!\!\!\!\!\!\!\!\!\!\!\!\!\!\!\!\!\!\!\!\!\!\!\!\!\!\!\!\!\!\!\!\! + \kappa^2 \sum_{l=1}^{L} \sum_{m=1}^{K} p_{lm} N  \lambda_{jjk}^2 \lambda_{jlm}^2 \tilde{\vect{x}}_{jk}^{\Htran}   \vect{D}_{\boldsymbol{\delta}(t)} \boldsymbol{\Psi}_j^{-1} \vect{X}_{lm}  \boldsymbol{\Psi}_j^{-1} \vect{D}_{\boldsymbol{\delta}(t)}^{\Htran} \tilde{\vect{x}}_{jk}. \notag
\end{align}
\end{theorem}
\begin{proof}
The expectations \eqref{eq:MRC-squared-norm}--\eqref{eq:MRC-cross-moment} are straightforward to compute, but the derivations are omitted due to the space limitations.
\end{proof}

By substituting the expressions from Theorem \ref{theorem:MRC-expectations} into \eqref{eq:achievable-SINR}, we obtain closed-form user rates that are achievable using MRC. The asymptotic behavior for large antenna arrays is now easily obtained.

\begin{corollary} \label{corollary:asymptotic-SINR}
If the MRC filter $\vect{v}_{jk}(t) =  \hat{\vect{h}}_{jjk}(t)$ is used, then
\begin{equation} \label{eq:asymptotic-SINR}
\mathrm{SINR}_{jk}(t) = \frac{ p_{jk} \lambda_{jjk}^2 \left(\tilde{\vect{x}}_{jk}^{\Htran}  \vect{D}_{\boldsymbol{\delta}(t)} \boldsymbol{\Psi}^{-1}_j \vect{D}_{\boldsymbol{\delta}(t)}^{\Htran} \tilde{\vect{x}}_{jk}\right)^2 }{\!\!\! \fracSum{(l,m) \neq (j,k)} \!\!\!  p_{lm} \lambda_{jlm}^2 |\tilde{\vect{x}}_{jk}^{\Htran}  \vect{D}_{\boldsymbol{\delta}(t)} \boldsymbol{\Psi}_j^{-1}  \vect{D}_{\boldsymbol{\delta}(t)}^{\Htran} \tilde{\vect{x}}_{lm} |^2 \!+\! \mathcal{O}(\frac{1}{N}) }
\end{equation}
where $\mathcal{O}(\frac{1}{N})$ denotes terms that go to $0$ as $\frac{1}{N}$ or faster as $N \rightarrow \infty$.
\end{corollary}
\begin{proof}
This is achieved by dividing all the terms in $\mathrm{SINR}_{jk}(t)$ by $\frac{1}{\lambda_{jjk}^2  N^2}$ and inspecting the asymptotics using Theorem \ref{theorem:MRC-expectations}.
\end{proof}

This corollary shows that the distortion noise and receiver noise vanish as $N \rightarrow \infty$, while the phase-drifts only has a minor asymptotic impact since the numerator and denominator of the SINR in \eqref{eq:asymptotic-SINR} are scaled symmetrically by $\vect{D}_{\boldsymbol{\delta}(t)}$. The terms that remain in the denominator depend on the pilot sequences $\tilde{\vect{x}}_{lm}$; hence, these terms are due to pilot contamination (PC) \cite{Marzetta2010a}; that is, inter-user interference in the estimation phase. Intra-cell PC is typically removed by making the pilot sequences orthogonal in space (i.e., $\tilde{\vect{x}}_{jk}^{\Htran} \tilde{\vect{x}}_{jm} = 0$ for $k \neq m$), which can be achieved by using the columns of a DFT matrix as pilot sequences \cite{Biguesh2004a}. Unfortunately, the phase-drifts caused by hardware imperfections break any spatial pilot orthogonality. Therefore, the only way to remove the intra-cell PC is to assign temporally orthogonal sequences within each cell (e.g., $x_{jk}(k) = \sqrt{p_{jk}}$ and $x_{jk}(t) = 0$ for $t\neq k$). Since temporal orthogonality reduces the total pilot power per user, $\| \tilde{\vect{x}}_{jk} \|^2$, by $1/K$, the simulations in Section \ref{sec:numerical-results} reveal that it is only beneficial for extremely large arrays.
Inter-cell PC can generally not be removed because there are only $B$ orthogonal sequences in the whole network.

\vspace*{-2mm}

\subsection{Scaling Laws on Hardware Imperfections}

\vspace*{-1mm}

The asymptotic results in Corollary \ref{corollary:asymptotic-SINR} reveal that the detrimental impact of hardware imperfections vanishes almost completely as $N$ grows large. This conclusion holds for any fixed values of the parameters $\delta$, $\kappa$, and $\xi$. The next corollary shows that it also holds if the parameters are increased with $N$ in a certain way.

\begin{corollary} \label{cor:scaling-law}
Suppose the hardware imperfection parameters are replaced as $\kappa^2 \mapsto \kappa_{0}^2 N^{\tau_1}$, $\xi \mapsto \xi_{0} N^{\tau_2}$, and $\delta \mapsto \delta_{0} (1+ \log_e(N^{\tau_3}) )$, for some scaling parameters $\tau_1,\tau_2,\tau_3 \geq 0$ and some initial values $\kappa_0,\xi_0,\delta_0 \geq 0$. If \vskip-2mm
\begin{equation} \label{eq:scaling-law}
\max(\tau_1,\tau_2) + \frac{\delta_{0} (t-B)}{2}\tau_3 \leq \frac{1}{2},
\end{equation}
$\mathrm{SINR}_{jk}(t)$ with MRC converges to a non-zero limit as $N \rightarrow \infty$.
\end{corollary}
\begin{proof}
This is achieved by substituting the new parameters into the SINR in \eqref{eq:achievable-SINR}, multiplying all terms  by $1/N^{1-\tau_3 \delta_0 (t-B)}$, and showing that the signal part is non-zero and the denominator is bounded.
\end{proof}

The corollary proves that one can increase the hardware imperfections with the number of antennas. This is a very important result for practical deployments, because it indicates that one can make the cost scale with $N$ at a slower pace than linear by using cheaper hardware. This property has been conjectured in overview articles, such as \cite{Larsson2014a}, and was proved in \cite{Bjornson2014a} using a system model with only additive distortion noise. Corollary \ref{cor:scaling-law} shows explicitly that the conjecture holds also for multiplicative phase-drifts and noise amplifications.

Since Corollary \ref{cor:scaling-law} is derived for MRC, \eqref{eq:scaling-law} provides a \emph{sufficient} scaling condition for any other receive filter that performs better than MRC. The scaling law consists of two terms: $\max(\tau_1,\tau_2)$ and $\frac{\delta_0(t-B)}{2}\tau_3$. The first term $\max(\tau_1,\tau_2)$ shows that the additive distortion noise and noise amplification can be increased simultaneously and independently, while the sum of the two terms manifests a tradeoff between increasing hardware imperfections that cause additive and multiplicative distortions. The system is particularly vulnerable to phase-drifts due to its accumulation, as seen from the second term which increases with $t$ and from that $\delta$ can scale only logarithmically with $N$. We can accept larger variances if the coherence block $T$ is small, which is in line with the results in \cite{Pitarokoilis2012a,Pitarokoilis2014a}.

\section{Numerical Results}
\label{sec:numerical-results}

The analytic results are evaluated in a simulation scenario with 16 cells and wrap-around to avoid edge effects; see Fig.~\ref{figure_simulationscenario}. Each square cell is $250 \times  250$ meters and is divided into 8 virtual sectors, where each sector contains one uniformly distributed UE (with minimum distance $35$ meters). Each sector has an orthogonal pilot sequence, but the same pilot is reused in the corresponding sector of other cells.

The channel attenuations are based on the 3GPP propagation model in \cite{LTE2010b}: $\lambda_{jlk} = 10^{s_{jlk}-1.53}/d_{jlk}^{3.76}$ where $d_{jlk}$ is the distance in meters between BS $j$ and UE $k$ in cell $l$ and $s_{jlk} \sim \mathcal{N}(0,0.25)$ is a realization of the shadow-fading. The transmit powers are $p_{jk} = -47$ dBm/Hz, the thermal noise power is $\sigma^2 = -174$ dBm/Hz, $B=8$ is the pilot sequence length, and the coherence block is $T=500$.

\begin{figure}
\begin{center}
\includegraphics[width=\columnwidth]{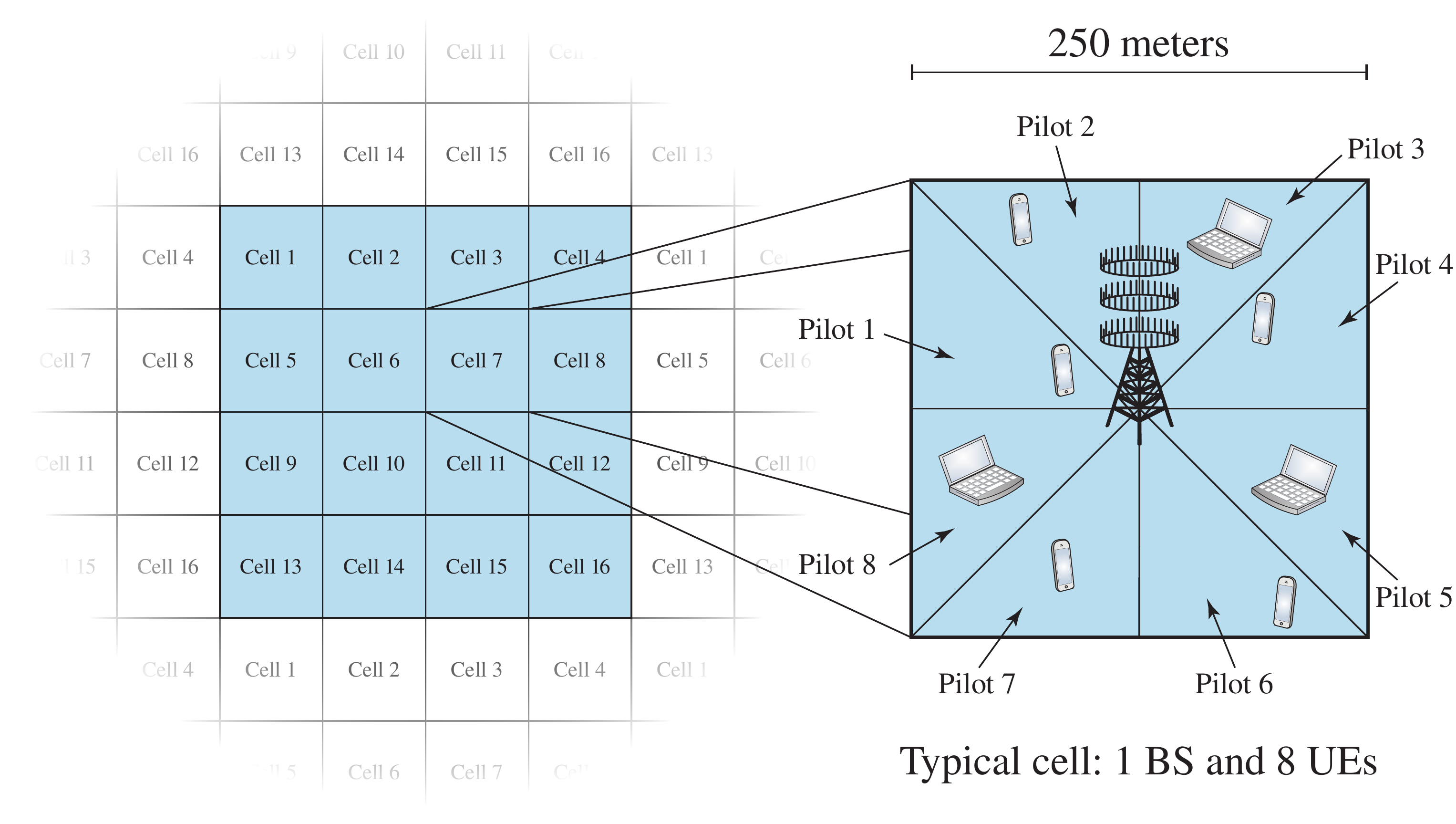}
\end{center}\vskip-6mm
\caption{The simulation scenario considers 16 square cells with wrap-around to avoid edge effects. Each cell is $250 \, \mathrm{m} \, \times \, 250 \, \mathrm{m}$ and consists of 8 UEs uniformly distributed in different parts of the cell.} \label{figure_simulationscenario} \vskip-2mm
\end{figure}

\begin{figure}[t!]
\begin{center}
\includegraphics[width=\columnwidth]{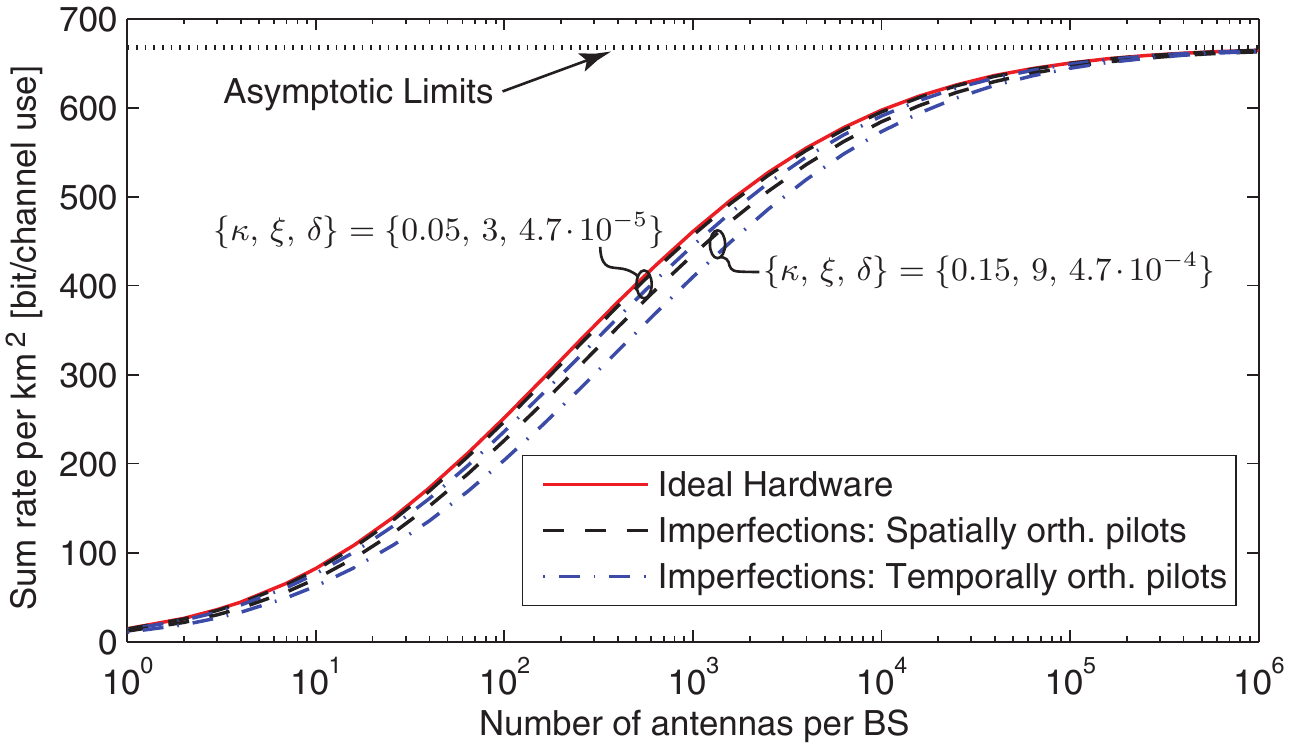}
\end{center}\vskip-6mm
\caption{Sum rate for different numbers of antennas, different hardware imperfections, and spatially or temporally orthogonal pilots.} \label{figure_simulationfigure1} \vskip-3mm
\end{figure}

We start by validating the asymptotic behaviors for fixed imperfections. Fig.~\ref{figure_simulationfigure1} shows the sum rate of all users (in the $1 \, \mathrm{km}^2$ area) as a function of the number of antennas $N$. The performance is given for ideal hardware and two types of hardware imperfections that are specified in the figure. The simulation shows that the convergence to the upper limit in Corollary \ref{corollary:asymptotic-SINR} is very slow---we used logarithmic scale on the horizontal axis because a million antennas is required for convergence. We observe that the sum rate reduces with hardware imperfections, but the loss is small and vanishes asymptotically.

Two types of pilot sequences are considered in Fig.~\ref{figure_simulationfigure1}: spatially orthogonal pilots selected from a DFT matrix \cite{Biguesh2004a} and temporally orthogonal pilots. As discussed in Section \ref{subsec:user-rates}, spatially orthogonal pilots is a better choice at practical $N$, although the limit in Corollary \ref{corollary:asymptotic-SINR} might be slightly larger for temporally orthogonal pilots.

Next, we focus on the practical range of $1 \leq N \leq 500$ in Fig.~\ref{figure_simulationfigure2}. We illustrate the scaling law from Corollary \ref{cor:scaling-law} by considering $\{ \kappa_0, \, \xi_0, \, \delta_0 \} \!=\! \{ 0.05, \, 3, \, 4.7 \cdot 10^{-5} \}$ and different $\tau_1,\tau_2,\tau_3$ which are specified in  Fig.~\ref{figure_simulationfigure2}. As expected, the combinations that satisfy the scaling law give minor performance losses, while the bottom curve goes to zero since the law is not fulfilled. The curves for MRC were generated using the analytical results of Section \ref{sec:performance-analysis} and match the marker symbols, which are the outputs of a Monte Carlo simulator.

The MRC filter was considered in Section \ref{sec:performance-analysis} since its low computational complexity is attractive for massive MIMO topologies. MRC provides a performance baseline for other receive filters which typically have higher complexity. In Fig.~\ref{figure_simulationfigure2} we also consider the filter \vskip-5mm
\begin{equation}
\vect{v}_{jk}^{\textrm{MMSE}}(t) \!=\!
\left( \sum_{l=1}^{L} \sum_{m=1}^{K} p_{lm} (\vect{G}_{jlm} \!+\! \kappa^2 \vect{D}_{\vect{G}_{jlm}} )  \!+\! \sigma^2 \xi \vect{I}_M  \! \right)^{\!\!-1} \!\! \hat{\vect{h}}_{jjk}(t)
\end{equation}
where $\vect{G}_{jlm} = \hat{\vect{h}}_{jlm}(t)  \hat{\vect{h}}_{jlm}^{\Htran}(t) + \vect{C}_{jlm}$ and $\vect{D}_{\vect{G}_{jlm}}$ is a diagonal matrix where the diagonal elements are the same as in $ \vect{G}_{jlm}$. This is an approximate minimum MSE (MMSE) filter that maximizes \eqref{eq:achievable-SINR} for a fixed channel realization. As seen from Fig.~\ref{figure_simulationfigure2}, the MMSE filter provides higher performance than the MRC filter. Interestingly, the losses due to hardware imperfections are similar but are somewhat larger for MMSE filters. This is because the MMSE filter exploits spatial interference suppression which is sensitive to imperfections.

\begin{figure}[t!]
\begin{center}
\includegraphics[width=\columnwidth]{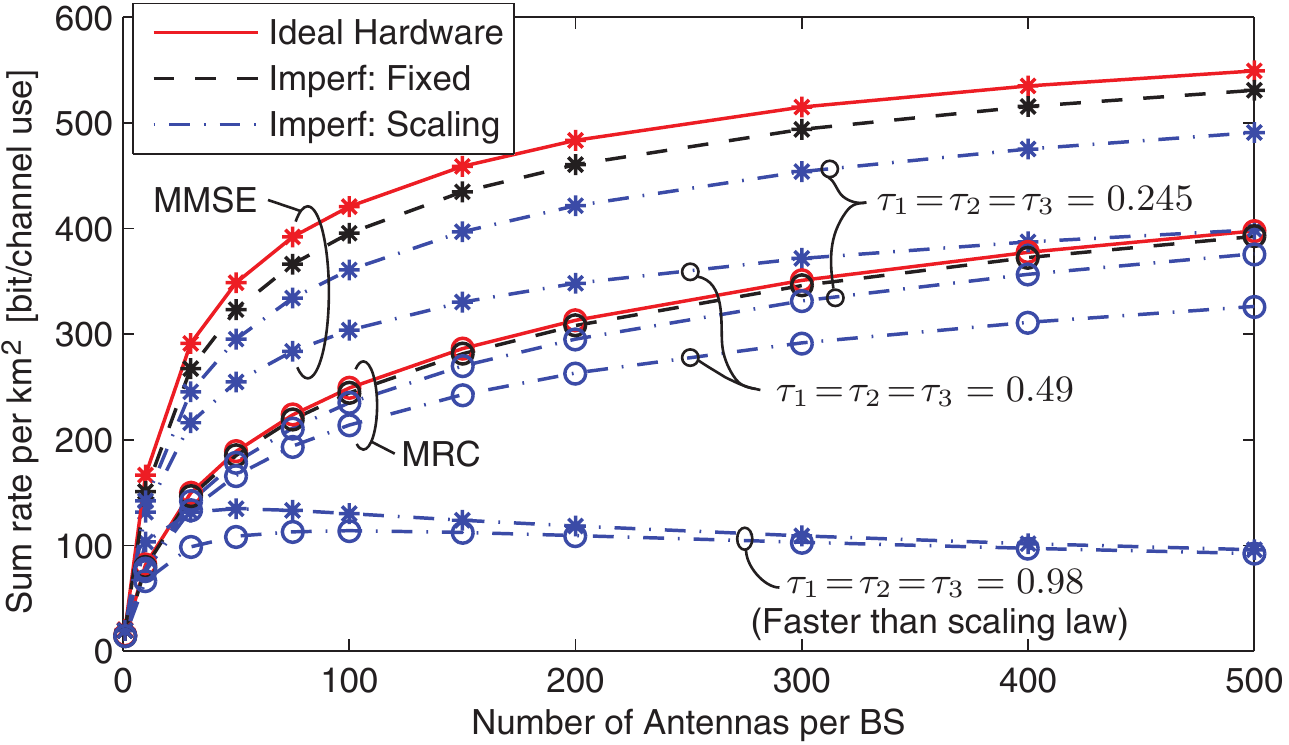}
\end{center}\vskip-6mm
\caption{Sum rate with MRC and MMSE filtering with ideal hardware, fixed imperfections, and imperfections that increase with $N$.} \label{figure_simulationfigure2} \vskip-4mm
\end{figure}

\section{Conclusion}

\vskip-1mm

A prerequisite for practical deployment of massive MIMO systems is that each antenna element in the large BS arrays is manufactured using low-cost components, which unfortunately are prone to hardware imperfections. In this work, we derived a scaling law that proves that massive MIMO systems are robust to hardware imperfections. This is a property that has been conjectured in prior works but only proved for simplified channel models with only additive distortion noise. We considered a more accurate uplink model with multiplicative phase-drifts, additive distortion noise, and noise amplifications. We derived an LMMSE channel estimator and the achievable user rates under MRC. Based on this model, our closed-form scaling law manifests how fast the hardware imperfections can increase with $N$, if non-zero user rates should be achieved. The simulation validates that the rate losses are small as compared to having ideal hardware. The scaling law reveals that the variance of the distortion noise and receiver noise can increase simultaneously as $\sqrt{N}$, but the scaling should be slower if also the phase-drift variance increases with $N$ (it can only increase logarithmically). Interestingly, the scaling results hold for other receive filters, such as the approximate MMSE filter.

\newpage

\bibliographystyle{IEEEbib}
\bibliography{IEEEabrv,refs}

\end{document}